\documentclass[11pt]{article}
\usepackage{amsmath, mathtools}
\usepackage{amsthm}
\usepackage{amssymb}
\usepackage{graphicx}
\usepackage{appendix}
\usepackage{color}

\usepackage{hyperref}

\usepackage{float}

\usepackage{multirow}

\newtheorem{theorem}{Theorem}[section]
\newtheorem{definition}[theorem]{Definition}
\newtheorem{assumption}[theorem]{Assumption}
\newtheorem{lemma}[theorem]{Lemma}

\newtheorem{remark}[theorem]{Remark}
\numberwithin{equation}{section}

\newcommand{\R}{\mathbb{R}}

\begin{document}

\title{\textsc{Sequential hypothesis testing in machine learning, and crude oil price jump size detection}}
\author{ Michael Roberts\footnote{Email: michael.roberts.1@ndus.edu}, Indranil SenGupta\footnote{Email: indranil.sengupta@ndsu.edu} \\ Department of Mathematics \\ North Dakota State University \\ Fargo, North Dakota, USA.}
\date{\today}
\maketitle

\begin{abstract}

In this paper, we present a sequential hypothesis test for the detection of the distribution of jump size in L\'evy processes. Infinitesimal generators for the corresponding log-likelihood ratios are presented and analyzed. Bounds for infinitesimal generators in terms of super-solutions and sub-solutions are computed. This is shown to be implementable in relation to various classification problems for a crude oil price data set. Machine and deep learning algorithms are implemented to extract a specific deterministic component from the data set, and the deterministic component is implemented to improve the Barndorff-Nielsen  \& Shephard  model, a commonly used stochastic model for derivative and commodity market analysis.

\end{abstract}



\section{Introduction}
\label{ch:newintro}

Various existing hedging algorithms and insurance risks depend on the underlying statistical model of the commodity market. Consequently, an improvement in the underlying model directly improves the hedging strategies and the understanding of insurance risks. In this paper,  we develop a novel statistical methodology for the refinement of stochastic models using various machine and deep learning algorithms.

As availability of information to the public through alternative data sources increases, machine learning is necessary for adequate analysis.  Currently, 97\% of North American businesses are using machine learning capabilities to analyze and apply data sources to their trading platforms and analytic focused activities (see \cite{refinitiv}).  The advent of these technologies allows participants to train, test, and project models using data that have historically been inaccessible.  ``Any innovation that makes better use of data, and enables data scientists to combine disparate sources of data in a meaningful fashion, offers the potential to gain competitive advantage"  (see \cite{refinitiv}). Trading capabilities, scale, scope, and speeds have increased exponentially with advancements in applications of Artificial Intelligence and Algorithmic trading.

A commonly used stochastic model for derivative and commodity market analysis is the Barndorff-Nielsen \& Shephard (BN-S) model (see \cite{BN1, BN-S1, BN-S2, BJS, Semere, Issaka, tatj}). In \cite{SWW}, the BN-S model is implemented to find an optimal hedging strategy for the oil commodity from the Bakken, a new region of oil extraction that is benefiting from fracking technology. In \cite{Wil2}, the BN-S model is used in this way, in the presence of quantity risk for oil produced in that region. 
In the recent paper \cite{recent}, a machine learning-based improvement of the BN-S model is proposed. It is shown that this refined BN-S model is more efficient and has fewer parameters than other models which are used in practice as improvements of the BN-S model.  Machine learning-based techniques are implemented for extracting a \emph{deterministic component}  ($\theta$) out of processes that are usually considered to be completely stochastic. Equipped with the aforementioned $\theta$, the obtained refined BN-S stochastic model can be implemented to incorporate \emph{long-range dependence} without actually changing the model. 

It is clear that the real challenge is to obtain an estimation of the value of the \emph{deterministic component} for an empirical data set. In \cite{recent}, a na\"ive way to find this value for crude oil price is proposed. The method proposed in that paper provides an algorithm to form a classification problem for the data set. After that, various machine and deep learning techniques are implemented for that classification problem. 

The primary motivation for this paper is the fact that the refined BN-S model can be successfully implemented to the analysis of crude oil price. In addition, it seems reasonable that some parameters of the refined BN-S model can be estimated by using various machine/deep learning algorithms. Consequently, it opens up the scope of an abundance of financial applications of the model to the commodity markets. With this motivation, in this paper, we investigate the problem from the perspective of sequential hypothesis testing. As described in \cite{Weld}, a sequential test of a hypothesis means any statistical test that gives a specific rule, at any stage of the experiment, for making one of the three decisions: (1) to accept the null hypothesis $H_0$, (2) to reject $H_0$, (3) to continue the experiment by making additional observation. Sequential hypothesis testing has many applications (see \cite{Baum, Brodsky, Dayanik, Golubev}). In the paper \cite{O}, the problem of testing four hypotheses on two streams of observations is examined.  A minimization result is obtained for the sampling time subject to error probabilities for distinguishing sequentially a standard versus a drifted two-dimensional Brownian motion. This result is further generalized in \cite{Roberts}, where the testing of four hypotheses on two streams of observations that are driven by L\'evy processes is presented. Consequently, the results in \cite{Roberts} are applicable for sequential decision making on the state of two-sensor systems. In one case, each sensor receives a L\'evy process with a drift term or no drift term. For the other case, each sensor receives data driven by L\'evy processes with large or small jumps. 
In this paper we show that a sequential test of a hypothesis can be implemented in relation to various classification problems for an empirical data set. Subsequently, machine and deep learning algorithms can be implemented to extract a  \emph{deterministic component} from a financial data set.

The organization of the paper is as follows. In Section \ref{sec12}, a refined BN-S model with some of its properties is presented. In Section \ref{ch:generalization}, we provide a general jump size detection analysis based on the sequential testing of hypotheses. In Section \ref{dataanalysissetup}, an overview of the data set is provided, and then two procedures in the predictive classification problem are introduced. Numerical results are shown in Section \ref{numericalresults}, and finally, a brief conclusion is provided in Section \ref{conclusion}. 

\section{A refined Barndorff-Nielsen \& Shephard model}
\label{sec12}

Many models in recent literature try to capture the stochastic behavior of time series. As an example, for the Barndorff-Nielsen \& Shephard model (BN-S model), the stock or commodity with price $S= (S_t)_{t \geq 0}$ on some filtered probability space $(\Omega, \mathcal{G}, (\mathcal{G}_t)_{0 \leq t \leq T}, \mathbb{P})$ is modeled by

\begin{equation}
\label{1}
S_t= S_0 \exp (X_t),
\end{equation}
\begin{equation}
\label{2}
dX_t = (\mu + \beta \sigma_t ^2 )\,dt + \sigma_t\, dW_t + \rho \,dZ_{\lambda t}, 
\end{equation}
\begin{equation}
\label{3}
d\sigma_t ^2 = -\lambda \sigma_t^2 \,dt + dZ_{\lambda t}, \quad \sigma_0^2 >0,
\end{equation}
where the parameters $\mu, \beta, \rho, \lambda \in \mathbb{R}$ with $\lambda >0$ and $\rho \leq 0$ and $r$ is the risk free interest rate where a stock or commodity is traded up to a fixed horizon date $T$.  In the above model $W_t$ is a Brownian motion, and the process $Z_{\lambda t}$ is a subordinator. Also $W$ and $Z$ are assumed to be independent, and $(\mathcal{G}_t)$ is assumed to be the usual augmentation of the filtration generated by the pair $(W, Z)$. 

However, the results and theoretical framework are far from being satisfactory.  The BN-S model does not incorporate the \emph{long-range dependence} property. As such, the model fails significantly for a longer ranges of time. To incorporate  \emph{long-range dependence}, a class of superpositions of Ornstein-Uhlenbeck (OU)-type processes is constructed in literature in terms of integrals with respect to independently scattered random measures (see \cite{BN1, Sem2}). With appropriate conditions, the resulting processes are incorporated with \emph{long-range dependence}. A limiting procedure results in processes that are second-order self-similar with stationary increments. Other resulting limiting processes are stable and self-similar with stationary increments. However, it is statistically unappealing to fit a large number of OU processes, at least by any formal likelihood-based method. To address this issue, in \cite{recent} a new method is developed. 

As proposed in \cite{recent}, $S= (S_t)_{t \geq 0}$ on some filtered probability space $(\Omega, \mathcal{F}, (\mathcal{F}_t)_{0 \leq t \leq T}, \mathbb{P})$, is given by \eqref{1}, where the dynamics of $X_t$ in \eqref{2} is given by
\begin{equation}
\label{2new}
dX_t = (\mu + \beta \sigma_t ^2 )\,dt + \sigma_t\, dW_t + \rho\left( (1-\theta) \,dZ_{\lambda t}+ \theta dZ^{(b)}_{\lambda t}\right), 
\end{equation}
where $Z$ and $Z^{(b)}$ are two independent subordinators, and $\theta \in [0,1]$ is a deterministic parameter. Machine learning algorithms can be implemented to determine the value of $\theta$. 
The process $Z^{(b)}$ in \eqref{2new} is a subordinator that has greater intensity than the subordinator $Z$. Also, $W$, $Z$ and $Z^{(b)}$ are assumed to be independent, and $(\mathcal{F}_t)$ is assumed to be the usual augmentation of the filtration generated by $(W, Z, Z^{(b)} )$.

In this case \eqref{3} will be given by
\begin{equation}
\label{4new}
d\sigma_t ^2 = -\lambda \sigma_t^2 \,dt + (1- \theta') dZ_{\lambda t} + \theta' dZ^{(b)}_{\lambda t} , \quad \sigma_0^2 >0,
\end{equation}
where, as before, $\theta' \in [0,1]$ is deterministic. It is worth noting that when $\theta=0$,  \eqref{2new} reduces to \eqref{2}. Similarly, when $\theta'=0$, \eqref{4new} reduces to \eqref{3}. 

We conclude this section with some properties of this new model. Note that $(1-\mu) \,dZ_{\lambda t}+ \mu dZ^{(b)}_{\lambda t}$, where $\mu \in [0,1]$, is also a L\'evy subordinator that is positively correlated with both $Z$ and $Z^{(b)}$. Note that the solution of \eqref{4new} can be explicitly written as
\begin{equation}
\label{5}
\sigma_t^2= e^{-\lambda t}\sigma_0^2 + \int_0^t e^{-\lambda (t-s)}\, \left( (1-\theta') dZ_{\lambda t} + \theta' dZ^{(b)}_{\lambda t}\right).
\end{equation}
The \emph{integrated variance} over the time period $[t, T]$ is given by $\sigma_{I}^{2}= \int_t^T \sigma_s^2\, ds$, and a straight-forward calculation shows 
\begin{equation}
\label{6}
\sigma_{I}^{2}= \epsilon(t,T) \sigma_t^2 + \int_t^T \epsilon(s,T)\, \left((1-\theta') dZ_{\lambda t} +  \theta' dZ^{(b)}_{\lambda t}\right),
\end{equation}
where 
\begin{equation}
\label{95}
\epsilon(s,T)= \left(1- \exp(-\lambda(T-s))\right)/\lambda, \quad t \leq s \leq T.
\end{equation}
We derive a general expression for the characteristic function of the conditional distribution of the log-asset price process appearing in the BN-S model given by equations \eqref{1}, \eqref{2new} and \eqref{4new}. For simplicity, we assume $$\theta=\theta'.$$

As shown in \cite{recent}, the advantages of the dynamics given by \eqref{1}, \eqref{2new}, and \eqref{4new} over the existing models are significant. The following theorem is proved in \cite{recent}. From this result, it is clear that as $\theta$ is constantly adjusted, for a fixed $s$, the value of $t$ always has an upper limit.   Consequently, $\text{Corr}(X_t, X_s)$ never becomes very small, and thus \emph{long-range dependence} is incorporated in the model.
\begin{theorem}
\label{big12222}
If the jump measures associated with the subordinators $Z$ and $Z^{(b)}$ are $J_Z$ and $J^{(b)}_Z$ respectively, and $J(s)= \int_0^s \int_{\mathbb{R}^+} J_Z(\lambda d\tau, dy)$, $J^{(b)}(s)= \int_0^s \int_{\mathbb{R}^+} J^{(b)}_Z(\lambda d\tau, dy)$; then for the log-return of the improved BN-S model given by \eqref{1}, \eqref{2new}, and \eqref{4new},
\begin{align}
\label{corrBNSimproved}
\text{Corr}(X_t, X_s)= \frac{\int_0^s \sigma_{\tau}^2 d\tau + \rho^2 (1-\theta)^2 J(s) + \rho^2 \theta^2 J^{(b)}(s)}{ \sqrt{\alpha(t) \alpha(s)}},
\end{align}
for $t>s$, where
$\alpha(\nu) = \int_0^{\nu} \sigma_{\tau}^2 d\tau + \nu\rho^2 \lambda ((1-\theta)^2 \text{Var}(Z_1)+  \theta^2 \text{Var}(Z^{(b)}_1)) $.
\end{theorem}

We implement the above analysis to empirical data sets. For example, we consider the West Texas Intermediate (WTI or NYMEX) crude oil prices data set for the period June 1, 2009 to May 30, 2019 (Figure 1). In the recent paper \cite{recent}, the appropriateness of modeling such data with a BN-S type stochastic volatility model is discussed. It is clear that such a process is dependent on random shocks, and thus an implementation of the classical model is argued in \cite{recent}. However, in Figure 2, we provide the autocorrelation function of the given data set. It is clear that the long-range dependence criteria must be incorporated in the stochastic model. This justifies the implementation of the refined BN-S model presented in this section. We will discuss a detailed data analysis in Section \ref{dataanalysissetup}. The implementation of the refined BN-S model in lieu of the classical BN-S model comes with the price of the estimation of $\theta$ as described earlier. In the later sections, this serves as a motivation to apply sequential hypothesis testing combined with various machine/deep learning algorithms. This leads to some novel numerical results  related to the present data set.

\begin{figure}[H]
\centering
\caption{Crude oil close price.}
\includegraphics[scale=.6]{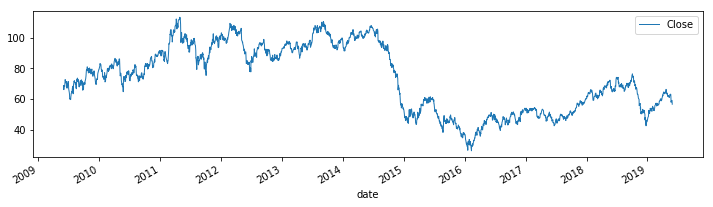}
\end{figure}

\begin{figure}[H]
\centering
\caption{Autocorrelation in crude oil close price.}
\includegraphics[scale=.6]{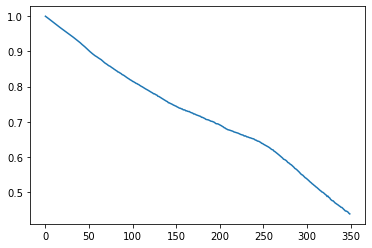}
\end{figure}

We denote $Z^{(e)}= (1- \theta) Z +  \theta Z^{(b)}$. Note that $Z^{(e)}$ is also a subordinator. We call this the \emph{effective} subordinator. We denote the cumulant transforms as $\kappa^{(e)}(\theta)= \log E^{\mathbb{P}}[e^{\theta Z_1^{(e)}}] $.  In this work, we make the following assumption similar to \cite{NV, ijtaf}.

\begin{assumption}
\label{a2}
Assume that $\hat{\theta}^{(e)}= \sup\{ \theta \in \mathbb{R} : \kappa^{(e)}(\theta) < + \infty \}>0$.
\end{assumption}

We state the following well-known result from \cite{NV, ijtaf} and denote the real part and imaginary part of $z \in \mathbb{C}$ as $\Re(z)$ and $\Im(z)$, respectively. 
\begin{theorem}
Let $Z$ be a subordinator with cumulant transform $\kappa$, and let $f: \mathbb{R}_{+} \to \mathbb{C}$ be a complex-valued, left continuous function such that $\Re(f) \leq 0$. Then
\begin{equation}
\label{8}
E\left[ \exp\left( \int_0^t f(s)\, dZ_{\lambda s}\right) \right]= \exp \left( \lambda \int_0^t \kappa (f(s))\, ds \right). 
\end{equation}
The above formula still holds if $Z= Z^{(e)}$ satisfies Assumption \ref{a2} and $f$ is such that $\Re(f) \leq \frac{\hat{\theta}^{(e)}}{(1+\epsilon)}$, for $\epsilon >0$.
\label{le1}
\end{theorem}

The \emph{Laplace transform} of $X_{T|t}$, the conditional distribution of $X_T$ given the information up to time $t \leq T$, is given by $\phi(z)= E^{\mathbb{P}}[\exp(zX_T)| \mathcal{F}_t]$, for $z \in \mathbb{C}$ such that the expectation is well-defined.

\begin{theorem}
\label{th1}
In the case of the general BN-S model described in equations \eqref{1}, \eqref{2new} and \eqref{4new}, the Laplace transform $\phi(z)= E[\exp(zX_T)| \mathcal{F}_t]$ of $X_{T|t}$ is given by
\begin{equation}
\label{91}
\phi(z)= \exp \left(z(X_t+ \mu(T-t))+ \frac{1}{2}(z^2+ 2\beta z) \epsilon(t, T) \sigma_t^2+ \lambda \int_t^T G(s,z)\, ds \right),
\end{equation}
where $G(s,z)= \kappa^{(e)}\left(\rho z + \frac{1}{2}(z^2+ 2 \beta z) \epsilon(s, T) \right)$.

The transform $\phi(z)$ is well defined in the open strip $\mathcal{S}= \{ z\in \mathbb{C} : \Re (z) \in (\theta_{-}, \theta_{+})\}$, where 
$$\theta_{-} = \sup_{t \leq s \leq T} \{-\beta - \frac{\rho}{\epsilon(s, T)}- \sqrt{\Delta_1}\},$$
and 
$$\theta_{+} = \inf_{t \leq s \leq T} \{-\beta - \frac{\rho}{\epsilon(s, T)} + \sqrt{\Delta_1}\},$$
where $\Delta_1= (\beta + \frac{\rho}{\epsilon(s, T)})^2 + 2 \frac{\hat{\theta}^{(e)}}{\epsilon(s,T)}$.
\end{theorem}

\begin{proof}
 The proof is in the appendix.
\end{proof}

\section{Jump size detection based on sequential hypothesis tests}
\label{ch:generalization}

In Section \ref{sec12}, it is observed that the refined BN-S model can be successfully implementable only when $\theta$ can be successfully computed for \eqref{2new} and \eqref{4new} (with $\theta= \theta'$). Note that, as discussed in the previous section, the value of $\theta$ is in the interval  $0 \leq \theta \leq 1$. However, in order to simplify the subsequent analysis, $\theta$ is rounded to either $0$ or $1$. This is motivated by the simplistic assumption that the jumps are either from one distribution or another. Also, in this case it is easier to interpret the confusion matrix corresponding to related classification algorithm. To find $\theta$, in \cite{recent}, a machine learning based empirical analysis is implemented. However, the procedure implemented in that paper does not  incorporate any hypothesis testing for $\theta$. In this section, we provide a more theoretical jump size detection analysis based on the sequential test of a hypothesis.

\subsection{Theoretical results}

We consider a L\'evy process $Z$ defined by L\'evy triplet $(\mu, \sigma^2, \nu^*)$, where $\mu$ is the drift, $\sigma$ is the diffusion, and $\nu^*(dx)=(1+\alpha x)\nu(dx)$ for some L\'evy measure $\nu$ defined on $\mathbb{R}^+$. We are interested in detecting a \emph{significant jump} in the process. Consequently, we wish to test the hypotheses
\begin{align}\label{jumphyp}
H_{0}: \alpha=0, \quad \quad H_{1}: \alpha=a> 0,  
\end{align}
which clearly address the size of the jumps in the L\'evy process.

The L\'evy process generates a filtration, which will be denoted $\mathcal{F}_t^{(i)}$, $i=0,1$. Further, the hypotheses induce probability measures $P_{i}$, $i=0,1$. We seek to create a decision rule $(\tau, \delta_\tau)$, where $\tau$ is a stopping rule with respect to $\mathcal{F}_t$, and $\delta_\tau$ is a random variable taking values in the index set $\{0, 1\}$, denoting which hypothesis to reject. 


Let  the log-likelihood ratio of the marginal density be given by
\begin{equation}\label{logratio}
u_t = \log \frac{dP_{0}}{dP_{1}},
\end{equation} 
and consider an interval $[l,r] \subset \mathbb{R}$ with $l < 0 < r$. We define the decision rules to be
\begin{align} \label{onedec}
\nonumber \tau&=\inf \{ t\geq 0: u_t \notin [l,r] \},\\
\nonumber \delta_{\tau} & =1, \hbox{ if } u_{\tau}\leq l, \\
\delta_{\tau} & = 0, \hbox{ if } u_{\tau}\geq r.
\end{align}

\begin{theorem}\label{jumpgen}
With the process $u_t$ defined as in \eqref{logratio}, we have infinitesimal generators, given by
$$\mathcal{L}\xi(x):= -\gamma \xi'(x) +\frac{1}{2} \beta^2 \xi''(x) +\int_{\mathbb{R}_+} \left(\xi(x+y)-\xi(x)-\frac{y \xi'(x)}{1+|y|}\right)K(dy),$$
for any suitable $\xi$, where 
\begin{align}
\beta&= -a\int_{x>0} (1\wedge x) \sigma^{-1}  x\nu(dx)\label{beta},\\
m&=a\int_{x>1} x\nu(dx)\label{m},\\
\gamma&=m-\frac{\beta^2}{2}+\int_0^1 (\log(1+x)^2-x)a\nu(dx),\label{gamma}\\
K&=a\log(1+x)\nu\label{K}. 
\end{align}
\end{theorem}
\begin{proof}
 The proof is in the appendix.
\end{proof}
Assign $\xi$ to be the probability of a correct decision given $H_0$. Then we have the partial integro-differential equations $\mathcal{L}\xi=0$ with boundary conditions
\begin{align}
\xi(l)=1, \quad \quad \quad \xi(r)=0.  \label{bcs}
\end{align}
Further, we have $ \xi>0$ inside $R=(l,r)$.

Before proving the existence of a solution to the stated boundary value problem, we need a few more definitions and a theorem from 
\cite{EU} that will be used:

\begin{definition}
An upper semicontinuous function $l: \R\to \R$ is a \emph{subsolution} of $$F(0,\xi,D\xi,D\xi^2,\mathcal{I}[\xi](x))=0$$ subject to boundary conditions \eqref{bcs} if, for any test function $\phi\in C^2(\R)$, at each maximum point $x_0 \in \bar{R}$ of $l-\phi$ in $B_\delta(x_0)$, we have 
$$E(l,\phi,x_0):=F(x_0,l(x_0),D\phi(x_0),D^2\phi(x_0), I_\delta^1[\phi](x_0)+I_\delta^2[l](x_0))\leq 0 \hbox{   if $x_0\in R$}$$
or
$$\min(E(l,\phi,x_0);u(x_0)-g(x_0))\leq 0 \hbox{  if $x_0\in \partial R$},$$
where
\begin{align*}
I_\delta^1[\phi](x_0)=\int_{|z|<\delta} \left( \phi(x_0+z)-\phi(x_0)-(D\phi(x_0) \cdot z)\textbf{1} _B (z)\right) d\mu_{x_0}(z), \\
I_\delta^2[u](x_0)=\int_{|z|\geq\delta} \left( u(x_0+z)-u(x_0)-(D\phi(x_0) \cdot z)\textbf{1}_B (z)\right) d\mu_{x_0}(z) .
\end{align*}

Similarly, a lower semicontinuous function $u: \R \to \R$ is a \emph{supersolution} of the same boundary value problem if for any test function $\phi\in C^2(\R)$, at each minimum point $x_0 \in \bar{R}$ of $u-\phi$ in $B_\delta(x_0)$, we have 
$$E(u,\phi,x_0)\geq 0\hbox{   if $x_0\in R$}$$
or 
$$\max(E(l,\phi,x_0);u(x_0)-g(x_0))\leq 0 \hbox{  if $x_0\in \partial R$}.$$
Finally, a \emph{viscosity solution} is a function whose upper and lower semicontinuous envelopes are respectively a \emph{sub-solution} and a \emph{super-solution}.
\end{definition}

\begin{theorem}\label{EUtheorem}
If $F:\R^5\to \R$, and
\begin{enumerate}
\item[(A1)] $F(x,u,p,X,i_1)\leq F(x,u,p,Y,i_2)$ if $X\geq Y$ and $i_1\geq i_2$,
\item[(A2)] there exists $\gamma>0$ such that for any $x, u, v, p, X, i \in \R,$
$$F(x,u,p,X,i)-F(x,v,p,X,i)\geq \gamma (u-v) \hbox{   if $u \geq v$},$$
for some $\epsilon>0$ and $r(\beta)\to 0$ as $\beta \to 0$, we have
$$F(y,v,\epsilon^{-1}(x-y),Y,i)-F(x,v,\epsilon^{-1}(x-y),X,i)\leq \omega_R (\epsilon^{-1}|x-y|^2+|x-y|+r(\beta)),$$
\item[(A3)] $F$ is uniformly continuous with respect to all arguments,\\
\item[(A4)] $\sup_{x\in \R} |F(x,0,0,0,0)|<\infty$,\\
\item[(A5)] $K$ is a L\'evy-It\^o measure,
\item[(A6)] the inequalities in \eqref{limsupinf} are strict,
\item[(A7)] for any $R>0$, there exists a modulus of continuity $\omega_R$ such that, for any $x,y\in\R$, $|v|\leq R$, $i\in \R$, and for any $X,Y\in \R$ satisfying
$$\left[\begin{array}{cc} X & 0 \\ 0 & Y\end{array}\right] \leq \frac{1}{\epsilon}\left[\begin{array}{cc} 1 & -1 \\ -1 & 1 \end{array}\right]+r(\beta)\left[\begin{array}{cc} 1 & 0 \\ 0 & 1 \end{array}\right],$$
\end{enumerate}
then there is a unique solution to $F(0,\xi,D \xi, D^2 \xi,\mathcal{I}[\xi](x))=0$ between any pair of super-solution and sub-solutions, defined below, where 
$$\mathcal{I}[\xi](x):=\int_{\mathbb{R}_+} \left(\xi(x+y)-\frac{y \xi'(x)}{1+|y|}\right)K(dy).$$
\end{theorem}

\begin{lemma}\label{EUlem}
In particular, the function $$F(x,u,p,X,i):=Mu+ \gamma p -\frac{\beta}{2}X - i$$ satisfies (A1)-(A4)  and our measure $K$ satisfies (A5) in \eqref{EUtheorem}, where 
$$M=\int_{\mathbb{R}_+}K(dy).$$
\end{lemma}
\begin{proof}
The proof is in the appendix.
\end{proof}

It is known that if $(X_t)_{t \geq 0}$ is a L\'evy process then there exists a unique c\'adl\'ag process $(Z_t)_{t \geq 0}$ such that 
$$dZ_t= Z_{t-}\, dX_t, \quad Z_0=1.$$
$Z$ is called the stochastic exponential or Dol\'eans-Dade exponential of $X$ and is denoted by $Z= \mathcal{E}(X)$. We now derive the infinitesimal generators. The results are motivated by \cite{Roberts}. Before proceeding, we present another formal definition:
\begin{definition}
We write that a function $f(x)=O(g(x))$  if we have some $M,\epsilon \in \R$ satisfying $|f(x)|\leq \epsilon g(x)$ for all $x>M$. Similarly, we write that a function $f(x)=o(g(x))$ if for any $\epsilon>0$, we have some $M \in \R$ satisfying $|f(x)|\leq \epsilon g(x)$ for all $x>M$.

The norm $\| f \|_\infty$ is defined as the essential supremum of the absolute value of $f$ over $\Omega$. It is the smallest number so that $\{ x : |f(x)|\geq \|f\|_\infty\}$ has measure zero.
\end{definition}

We state the additional limit assumptions on $F$ from \cite{EU}: 
\begin{align} \liminf_{y\to x, y\in \bar{\Omega}, \eta \downarrow 0, d(y)\eta^{-1} \to 0} \left[ \sup_{0<\delta \in [d(y),r)} \inf_{s\in [-R,R]} F(y,s,p_\eta(y),M_\eta(y),I_{\eta,\delta,r}(y))\right]&<0, \nonumber \\ 
\limsup_{y\to x, y\in \bar{\Omega}, \eta \downarrow 0, d(y)\eta^{-1} \to 0} \left[ \inf_{0<\delta \in [d(y),r)} \sup_{s\in [-R,R]} F(y,s,-p_\eta(y),-M_\eta(y),-I_{\eta,\delta,r}(y))\right]&<0, \label{limsupinf}\end{align}
where
\begin{equation*}
p_\eta(y) =O(\epsilon^{-1})+\frac{k_1+o(1)}{\eta} Dd(y),
\end{equation*}
\begin{equation*}
M_\eta(y)= O(\epsilon^{-1})+\frac{k_1+o(1)}{\eta} D^2d(y)-\frac{k_2+o(1)}{\eta^2} Dd(y) \otimes Dd(y),
\end{equation*}
\begin{align*}
I_{\eta,\delta,r}(y)&=-\nu I_{\delta,r}^{\hbox{ext},1}(y)+2\|u\|_\infty I_{\beta(\nu),r}^{\hbox{int},1}(y)\\
&-\frac{k_1+o(1)}{\eta} \left(I^{\hbox{tr}}(y)+I_{\beta(\eta),r}^{\hbox{int},2}(y)+I_{\delta,r}^{\hbox{ext},2}(y)-\|D^2d\|_\infty I^4_{\delta,\beta(\eta),r}(y)\right)\\
&+O(\epsilon^{-1})\left( 1+ o(1)I_{\beta(\eta),r}^{\hbox{int},3}(y)+o(1)I_{\delta,r}^{\hbox{ext},3}(y) \right),
\end{align*}
with $O(\epsilon^{-1})$ not depending on $k_1$ nor $k_2$, and
\begin{equation*}
\mathcal{A}_{\delta,\beta,r}(x) :=\{z\in B_r: -\delta\leq d(x+z)-d(x)\leq \beta\},
\end{equation*}
\begin{equation*}
\mathcal{A}^{\hbox{ext}}_{\delta,r}(x):=\{z\in B_r:  d(x+z)-d(x)< -\delta\},
\end{equation*}
\begin{equation*}
\mathcal{A}^{\hbox{int}}_{\beta,r}(x):=\{z\in B_r:  d(x+z)-d(x)>\beta\},
\end{equation*}
\begin{equation*}
I^{\hbox{ext},1}_{\delta,r}(x) := \int_{\mathcal{A}^{\hbox{ext}}_{\delta,r}(x)} d\mu_x(z),
\end{equation*}
\begin{equation*}
I^{\hbox{ext},2}_{\delta,r}(x) := \int_{\mathcal{A}^{\hbox{ext}}_{\delta,r}(x)} Dd(x)\cdot zd\mu_x(z),
\end{equation*}
\begin{equation*}
I^{\hbox{ext},3}_{\delta,r}(x) := \int_{\mathcal{A}^{\hbox{ext}}_{\delta,r}(x)}|z| d\mu_x(z),
\end{equation*}
\begin{equation*}
I^{\hbox{int},1}_{\beta,r}(x) := \int_{\mathcal{A}^{\hbox{ext}}_{\beta,r}(x)} d\mu_x(z),
\end{equation*}
\begin{equation*}
I^{\hbox{int},2}_{\beta,r}(x) := \int_{\mathcal{A}^{\hbox{ext}}_{\beta,r}(x)} Dd(x)\cdot zd\mu_x(z),
\end{equation*}
\begin{equation*}
I^{\hbox{int},3}_{\beta,r}(x) := \int_{\mathcal{A}^{\hbox{ext}}_{\beta,r}(x)}|z| d\mu_x(z),
\end{equation*}
\begin{equation*}
I^{4}_{\delta,\beta,r}(x) := \frac{1}{2}\int_{\mathcal{A}_{\delta,\beta,r}(x)}|z|^2 d\mu_x(z),
\end{equation*}
\begin{equation*}
I^{\hbox{tr}}(x) := \int_{r<|z|<1} Dd(x)\cdot zd\mu_x(z).
\end{equation*}

Using all of the previous, we can finally state the existence theorem.

\begin{theorem}\label{solution}
If $\xi$ is monotonic, then the partial integro-differential equation $\mathcal{L}\xi=0$, subject to boundary conditions \eqref{bcs} and  $\xi>0$ has a viscosity solution between sub-solution and super-solution
\begin{align*}
g(x)&= \exp(B(x-l))\frac{\sinh \left(\frac{r-x}{\beta} \sqrt{ 2M +B^2}\right)}{\sinh \left(\frac{r-l}{\beta} \sqrt{ 2M +B^2}\right)},\nonumber \\
f(x)&= \frac{\exp(2Br)-\exp(2Bx)}{\exp(2Br)-\exp(2Bl)},
\end{align*}
where
\begin{align*}
C&=\int_0^\infty \frac{y}{1+|y|}K(dy),\nonumber\\
B&= \frac{2 (C+\gamma) }{\beta^2},\nonumber\\ 
M&=\int_0^\infty K(dy).
\end{align*}
\end{theorem}
\begin{proof}
The proof is in the appendix.
\end{proof}

\begin{remark}
The existence of a more general viscosity solution to a higher dimensional problem is shown without assuming monotonicity in the paper \cite{Roberts}. The monotonicity assumption yields a tighter super- and sub-solution envelope and is here to make the application of this theorem to time series data more effective.
\end{remark}

\subsection{Jump size detection algorithm}
We will use the prior super- and sub-solutions as envelopes to approximate an important parameter in the following algorithm that uses the previous hypothesis test to classify L\'evy processes as having small or large jumps.

Given oil price close values in length-$n$ work day periods, we do the following: 

\begin{enumerate}
\item An inverse Gaussian density $\nu(dx)$ is fit to the distribution of negative percent daily jumps for the entire (training) data set. 

\item We then fit the density of the L\'evy measure from \ref{jumpgen}, $\nu^*(dx)=(1+ax)\nu(dx)$, to the distribution of the negative percent daily jumps for the $n$-length period. This gives a test statistic $a$ for the parameter in the hypothesis test.

\item We calculate the standard deviation $\sigma$ of all daily percent changes for the $n$-length period.

\item Using the density $\nu^*$ and standard deviation $\sigma$, we calculate $\gamma$, $\beta$, and $C$ from \ref{jumpgen} and \ref{solution}.

\item The left side of the interval is chosen to be $-1$, then using $a$, $\sigma$, $\beta$, $\gamma$, and $C$ in the super- and sub-solution equations in \ref{solution}, we can solve for the right side of the interval using $f(0)=1-\alpha_0$ and $g(0)=1-\alpha_0$, and take the average of the two. The parameter $\alpha_0$ is chosen to be the maximum desired probability of a Type-I Error.

\item Simulations of the log-likelihood process with drift $\gamma$, volatility $\beta$, and jumps represented by an inverse Gaussian process with expected value $-t \int_0^\infty x K(dx)$ at time $t$, are run. We record the frequency of exits out of the right-side of the interval to get a number that represents, relatively, the size of the jumps. We call this number the right-exit frequency.

\item Time periods whose right-exit frequencies are at or above a certain threshold $p^*$ are then classified as having large jumps, while the others are classified as having small jumps.

\end{enumerate}

\subsection{Effectiveness on simulated data}
To demonstate the capacity of the hypothesis testing algorithm in distinguishing between processes with small and large jumps, we run it on simulated data. Multiple classes of L\'evy processes are simulated, all of which start initially at $100$:
\begin{enumerate}
\item a training data time series with drift $1$, diffusion $0.5$, and jumps that follow an inverse Gaussian distribution with mean $1$ and scale factor $1$, which gives a L\'evy measure $\nu$,
\item a control data set of $100$ processes with parameters identical to the training data,
\item a data set of $100$ processes with obvious large negative jumps: the parameters are the same as the training set except the L\'evy measure is now represented by $(1+x)\nu(dx)$, and
\item a data set of $100$ processes with subtle large negative jumps: the drift is increased to $3$ compensate for the previous increase in jump size.
\end{enumerate}
The training time series is run for $500$ time periods, and the other three are run for $30$ each, with representatives shown in Figure \ref{simprocs}.

\begin{figure}
\centering
\caption{The training data and a representative from each other data set.}
\label{simprocs}
\includegraphics[scale=.35]{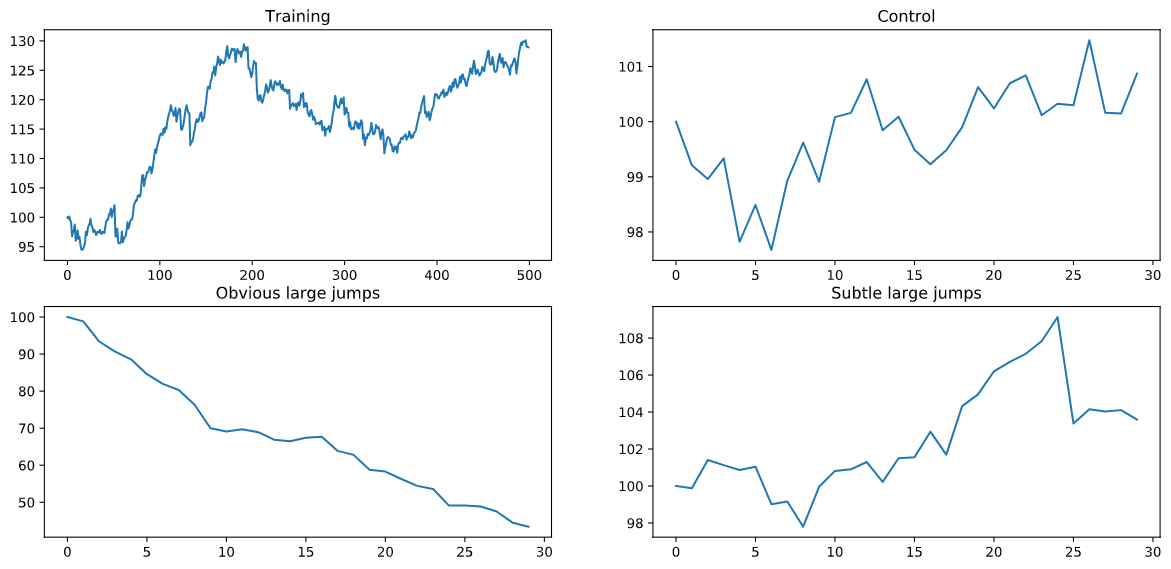}
\end{figure}

 The hypothesis test algorithm with $p^*=8$ and $\alpha_0=0.1$ is run on each data set. (The parameter $p^*$ is chosen here to gives desirable results and will be used in the application in the next section.) For the control, $79$ out of the $100$ processes are correctly identified as coming from the distribution with small jumps. This is to be expected because some number of processes would randomly have signficantly larger jumps just by chance. All $100$ from the obvious large jumps set are identified as having large jumps, and $85$ out of the $100$ processes in the subtle large jumps set are correctly identified.

Alternatively, a na\"ive approach of simply classifying each $30$ day period based on comparing only the mean jump size relative to the training data's mean jump size results in only $28$ of the $100$ control processes being correctly identified; although it correctly identifies all but $3$ of the large jump simulations. Because of the significant potential for Type-I Error in this na\"ive approach, the hypothesis test algorithm has evident advantages.

\section{Prediction method}
\label{dataanalysissetup}

We briefly discussed the data set in Section \ref{sec12}. In this section, we present an overview of the data set in its entirety, and then develop two procedures used in the predictive classification problem. As discussed in Section \ref{sec12}, we consider the West Texas Intermediate (WTI or NYMEX) crude oil prices data set for the period June 1, 2009 to May 30, 2019.  West Texas Intermediate  crude oil is described as light sweet oil traded and delivered at Cushing, Oklahoma. WTI usually refers to the price of the New York Mercantile Exchange (NYMEX) WTI Crude Oil futures contract. For WTI, spot and futures prices are used as a benchmark in oil pricing. The WTI crude oil futures contract specifies the deliverable asset for the contract to be a blend of crude oil, as long as it is of acceptable lightness and sweetness.  The data set is available online in \cite{datasource}. We index the available dates from 0 (for June 1, 2009) to 2529 (for  May 30, 2019).

The following table (Table 1) summarizes various estimates for the data set.

\begin{table}[H]
\centering
\caption{Properties of the empirical data set.}
  \begin{tabular}{ | l | c | r |}
    \hline
    & Daily Price Change & Daily Price Change \% \\ \hline
    Mean & -0.0047 & 0.01370 \% \\ \hline
    Median & 0.04399 &  0.06521 \%\\     \hline
Maximum & 7.62 &  12.32 \%\\     \hline
Minimum & -8.90 &  -10.53 \%\\     \hline
  \end{tabular}
\end{table}

In Figure \ref{closehist} the distribution plot for close oil price is provided.  Histograms for daily change in close oil price and daily change percentage in close oil price are provided in Figure \ref{closehistchange} and Figure \ref{dailypchist}, respectively, for exploratory purposes.

\begin{figure}[H]
\centering
\caption{Distribution plot for close oil price.}
\label{closehist}
\includegraphics[scale=.4]{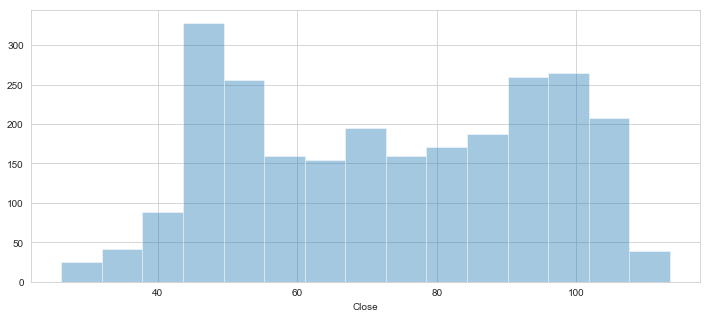}
\end{figure}

\begin{figure}[H]
\centering
\caption{Histogram for daily change in close oil price.}
\label{closehistchange}
\includegraphics[scale=.7]{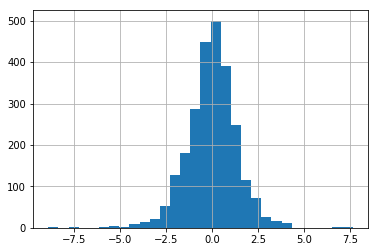}
\end{figure}
\begin{figure}[H]
\centering
\caption{Histogram for daily change percentage in close oil price.}
\label{dailypchist}
\includegraphics[scale=.7]{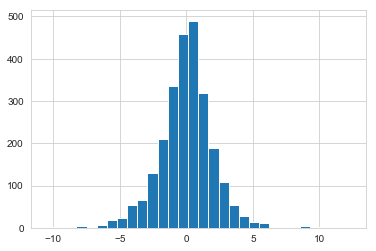}
\end{figure}

In the following subsections, two procedures are described for constructing the related classification problem. The procedures differ in the features used for the analysis:  percent daily changes and right-exit frequencies. In each, the algorithm at the end of Section \ref{ch:generalization} is used to determine whether an individual time period has large or small jumps, represented by the right-exit frequency of that time period. The machine learning algorithms are then used to predict whether the right-exit frequency of the next time period will be large or small. Consequently, before truncation, the resulting probabilities of large jumps from each machine learning algorithm can be used to update $\theta$ from the refined BN-S model in Section \ref{sec12} each period.

\subsection{Percent daily changes as features}\label{percentatt}
We implement the following procedure to create a machine learning classification problem:

\begin{enumerate}
\item We consider the percent daily changes for the historical oil price data and create a new data-frame from the old where the columns will be $n$ consecutive daily change percents. For example, if the changes are 
$$a_1,a_2,a_3,...,$$
then the first row of the data set will be 
$$a_1,a_2,...,a_n,$$
and the second row will be 
$$a_2,a_3,...,a_{n+1},$$
and so forth. \label{newdf}

\item We create a target column that is $0$ if the right-exit frequency of the next disjoint $n$ days is less than some threshold $p^*$, and is $1$ otherwise. For example, if the time period $$a_{1+n},a_{2+n},...,a_{2n-1}$$ has a significant frequency of right-exits, then the time period $$a_1,a_2,...,a_n$$ is given a target value $1$.

\item We run various classification algorithms where the input is a list of $n$ consecutive close prices, and the output is a $1$ to represent large jumps or $0$ to represent small jumps of the next $n$ consecutive close prices. Classification reports and confusion matrices are evaluated for each algorithm.

\end{enumerate}

\subsection{Right-exit frequencies as features}\label{rightatt}
We implement the following procedure to create a machine learning classification problem:

\begin{enumerate}
\item Similar to the previous, we consider close prices for the historical oil price data and create a new data-frame exactly as before. 

\item A new column is created that holds the right-exit frequencies for each consecutive set of $n$ days, say, 
$$b_1,b_2,b_3,....$$ 
These represent how large the jumps in close prices are for the previous $n$ days.

\item From this column, a new staggered data-frame is created, similar to before.

\item Finally, a target column is created: if the row is 
$$b_{30},b_{31},...,b_{30+n-1},$$
then the entry in the target column will be $b_{30+2n-1}$. This is the right-exit frequency of the next disjoint $n$-day period.

\item We run various classification algorithms where the input is a list of $n$ consecutive right-exit frequencies, and the output is $1$ to represent large jumps or $0$ to represent small jumps of the next $n$ consecutive close prices. Classification reports and confusion matrices are evaluated for each algorithm.
\end{enumerate}


\subsection{Numerical results}
\label{numericalresults}

Now we apply the procedures described in the last section to specific cases. For this section, the period length $n=30$. Further, $\alpha$, the parameter representing an approximation for the Type-I Error of the test is chosen to be $\alpha=0.9$, and because it worked optimally in the simulation study, the cut off for significant right-exit frequencies is chosen to be $p^*=8$. Two different time periods are used for training, and two are used for testing. The time periods are 
\begin{itemize}
\item $T_1$: \emph{training data (index)}: October 21, 2009 (100) to May 17, 2013 (1000); and \emph{testing data (index)}: April 21, 2017 (2000) to April 10, 2019 (2500); 
\item $T_2$: \emph{training data (index)}: August 11, 2009 (50) to May 13, 2013 (1500); and \emph{testing data (index)}: October 5, 2015 (1600) to January 29, 2019 (2450).

\end{itemize}

 Because the data is significantly imbalanced in favor of small jump time periods, random small jump periods from the training data are removed while performing algorithms \ref{percentatt} and \ref{rightatt}. This allows the neural nets and other algorithms to isolate the attributes of large and small jump periods without becoming distracted by the imbalanced frequency of small jump periods. Without doing so, the algorithms often predict all time periods to be small jump periods -- simply because those are more prevalent. The results of the machine learning algorithms using the time periods above are recorded in the following tables (Tables 2-5).  Those used are linear regression (LR), decision trees (DT), random forests (RF), and three different types of neural nets, (A) a standard net, (B) a long-short term memory net, and (C) a LSTM net with a batch normalizer.

Most of the machine learning algorithms perform better than how one might expect from guessing uniformly whether the next time period would have big jumps. Some perform notably poorly, however, particularly the LSTM neural nets without a batch normalizer. However, the neural nets with a batch normalizer consistently perform quite well.


Figure \ref{refhist} provides a histogram showing the distribution of right-exit frequencies for period lengths of 30 business days in the $T_2$ testing data. For each set of 30 consecutive days, $10$ simulations are run, and the frequency of right-exits is recorded. The $x$-axis in the figure is the number of simulated processes that exit to the right of the testing interval for a given period, while the $y$-axis is the number of 30 day periods with that frequency of right-exits.

Once the value of $\theta$ is estimated, this can be implemented in the refined BN-S model \eqref{2new} (and, \eqref{4new}, with $\theta=\theta'$). Equipped with $\theta$, as described in \cite{recent} and as shown in Theorem \ref{big12222}, the refined BN-S stochastic model can be used to incorporates \emph{long-range dependence} without actually changing the model. In addition, this shows a real-time application of data science for extracting a \emph{deterministic component} out of processes that are thus far considered to be completely stochastic. By the \emph{deterministic component}, it is meant that $\theta$ is a deterministic signal. This is deterministic in the sense that its value is extracted from the data before the model is implemented. Once the value of $\theta$ is obtained, it is kept constant for a certain period of time. For the computational effectiveness of $\theta$, the results in Tables 2-5 show better estimation compared to the benchmark study in \cite{recent}. 

\begin{table}[H]
\caption{Various estimations for $T_1$, using daily percent changes as features. }
  \begin{tabular}{ | c | c | c |c|c|c| c|}
   \hline
                                        & LR     & DT &    RF  & Neural Network (A) & LSTM (B) & BN (C) \\ \hline
    precision $\theta=0$ & 0.92 &  0.89 & 0.89  & 0.88 & 0.83 &  0.93 \\ \hline
    recall       $\theta=0$ & 0.61& 0.56  & 0.64 &0.54 & 0.77 & 0.88 \\     \hline
    f1-score  $\theta=0$ & 0.74 & 0.69  & 0.74 &0.67  &0.80 & 0.90 \\     \hline
    support   $\theta=0$ & 340  & 340   & 340  & 340   & 340 &  340  \\     \hline
    precision $\theta=1$ & 0.29 & 0.24 & 0.25 & 0.22 & 0.19 & 0.53  \\ \hline
    recall      $\theta=1$ & 0.76 & 0.67  & 0.60 & 0.64 & 0.26 & 0.66 \\     \hline
   f1-score  $\theta=1$ & 0.42 & 0.35  & 0.36 & 0.33 & 0.22 & 0.59 \\     \hline
   support   $\theta=1$ & 70 & 70   & 70  & 70  & 70 & 70 \\     \hline
\end{tabular}
\end{table}

\begin{table}[H]
\caption{Various estimations for $T_1$, using right-exit frequencies as features. }
  \begin{tabular}{ | c | c | c |c|c|c| c|}
   \hline
                                        & LR     & DT &    RF  & Neural Network (A) & LSTM (B) & BN (C) \\ \hline
    precision $\theta=0$ & 0.89 &  0.92 & 0.87  & 0.83 & 0.88 &  0.85 \\ \hline
    recall       $\theta=0$ & 0.66 & 0.67  & 0.69 &0.56 & 0.21 & 0.71 \\     \hline
    f1-score  $\theta=0$ & 0.76 & 0.77  & 0.77 &0.67  &0.34  & 0.77 \\     \hline
   support   $\theta=0$ & 340  & 340   & 340  & 340   & 340 &  340  \\     \hline
    precision $\theta=1$ & 0.24 & 0.30 & 0.25 & 0.17 & 0.18 & 0.21  \\ \hline
    recall      $\theta=1$ & 0.61 & 0.70  & 0.50 & 0.43 & 0.86 & 0.39 \\     \hline
   f1-score  $\theta=1$ & 0.37 & 0.42  & 0.33 & 0.24 & 0.30 & 0.27 \\     \hline
   support   $\theta=1$ & 70 & 70   & 70  & 70  & 70 & 70 \\     \hline
\end{tabular}
\end{table}

\begin{table}[H]
\caption{Various estimations for $T_2$, using daily percent changes as features. }
  \begin{tabular}{ | c | c | c |c|c|c| c|}
   \hline
                                        & LR     & DT &    RF  & Neural Network (A) & LSTM (B) & BN (C) \\ \hline
    precision $\theta=0$ & 0.79 &  0.74 & 0.79  & 0.80 & 0.77 &  0.75 \\ \hline
    recall       $\theta=0$ & 0.82 & 0.50  & 0.57 &0.66 & 0.58 & 0.91 \\     \hline
    f1-score  $\theta=0$ & 0.80 & 0.59  & 0.66 &0.72  &0.66 & 0.82 \\     \hline
    support   $\theta=0$ & 519& 519& 519& 519& 519& 519  \\     \hline
    precision $\theta=1$ & 0.57 & 0.37 & 0.42 & 0.47 & 0.41 & 0.65  \\ \hline
    recall      $\theta=1$ & 0.53 & 0.63  & 0.68 & 0.63 & 0.63 & 0.37 \\     \hline
   f1-score  $\theta=1$ & 0.55 & 0.46  & 0.52 & 0.54 & 0.50 & 0.47 \\     \hline
   support   $\theta=1$ &241&241&241&241&241&241 \\     \hline
\end{tabular}
\end{table}

\begin{table}[H]
\caption{Various estimations for $T_2$, using right-exit frequencies as features. }
  \begin{tabular}{ | c | c | c |c|c|c| c|}
   \hline
                                        & LR     & DT &    RF  & Neural Network (A) & LSTM (B) & BN (C) \\ \hline
    precision $\theta=0$ & 0.80 &  0.76 & 0.79  & 0.76 & 0.80 &  0.75 \\ \hline
    recall       $\theta=0$ & 0.54 & 0.56  & 0.58 &0.54 & 0.16 & 0.63\\     \hline
    f1-score  $\theta=0$ & 0.65 & 0.64  & 0.67 &0.63  &0.27  & 0.68 \\     \hline
   support   $\theta=0$ & 519& 519& 519& 519& 519& 519  \\     \hline
    precision $\theta=1$ & 0.42 & 0.39 & 0.42 & 0.39 & 0.34 & 0.40  \\ \hline
    recall      $\theta=1$ & 0.70& 0.62  & 0.66 & 0.64 & 0.91& 0.54 \\     \hline
   f1-score  $\theta=1$ & 0.52 & 0.48  & 0.52 & 0.49 & 0.49 & 0.46 \\     \hline
   support   $\theta=1$ &241&241&241&241&241&241 \\     \hline
\end{tabular}
\end{table}

\begin{figure}[H]
\centering
\caption{Histogram for daily (previous 30 days) right-exit frequencies.}
\label{refhist}
\includegraphics[scale=.6]{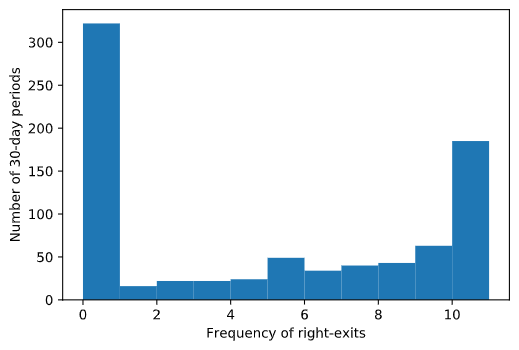}
\end{figure}

\section{Conclusion}
\label{conclusion}

Motivated by the fact that the refined BN-S model can be successfully implemented to the analysis of crude oil price, and that the parameters of the refined BN-S model can be estimated by using various machine/deep learning algorithms, in this paper we study the refined BN-S model from the sequential hypothesis testing perspective, with an application to the oil market.

Mathematical modeling of oil price data is directly inspired by various stochastic models. Thorough understanding and theoretical development of appropriate stochastic models contribute to a better understanding of the risk-management problem of various commodities, and various existing algorithms in a financial market depend on the underlying statistical model. Consequently, an improvement in the underlying model directly improves the existing algorithms. In this paper, a sequential decision making problem in connection to the L\'evy process is studied to analyze the jump size distribution. This is coupled with various machine and deep learning techniques to improve the existing stochastic models.  Consequently, the analysis presented in this paper provides a necessary mathematical framework for an appropriate generalization of various stochastic models.

Future works related to this topic should definitely include seeking to find a more adequate approximation for the right side of the decision rule interval. This would greatly increase the sensitivity of the algorithms in classifying large-jump time periods, thereby requiring less computational power for even better results. Applications to other data sets more independent of exogenous forces, and even across multiple streams of data using \cite{Roberts}, should also be explored. Finally, constructing decision rules for hypothesis tests on other parameters in the underlying processes could open up this type of analysis to more generalized scenarios. \\

\section{Appendix}
Proof of Theorem \ref{th1}:
\begin{proof}
We obtain from equation \eqref{2new} 
$$X_T= \zeta + \beta \sigma_{I}^{2}+ \int_t^T \sigma_s\, dW_s + \rho \int_t^T dZ^{(e)}_{\lambda s},$$
where $\zeta= X_t + \mu(T-t)$. Let $\mathcal{G}$ denote the $\sigma$-algebra generated by $Z^{(e)}$ up to time $T$ and by $\mathcal{F}_t$. Then, proceeding by iterated conditional expectations, we obtain
\begin{align*}
\phi(z) & = E^{\mathbb{P}}[\exp(zX_T)| \mathcal{F}_t] \\
& = E^{\mathbb{P}}\left[ E^{\mathbb{P}}\left[\exp(z(\zeta + \beta \sigma_{I}^{2}+ \int_t^T \sigma_s\, dW_s + \rho \int_t^T dZ^{(e)})) | \mathcal{G}\right] | \mathcal{F}_t \right] \\
& = E^{\mathbb{P}} \left[ \exp (z(\zeta + \beta \sigma_{I}^{2} + \rho \int_t^Td Z^{(e)}_{\lambda s})) E^{\mathbb{P}}\left[\exp(z \int_t^T \sigma_s\, dW_s)| \mathcal{G} \right] | \mathcal{F}_t\right] \\
& = E^{\mathbb{P}}\left[ \exp\left(z(\zeta + \beta \sigma_{I}^{2} + \rho \int_t^T dZ^{(e)}_{\lambda s}) + \frac{1}{2}\sigma_{I}^{2}  z^2 \right) | \mathcal{F}_t \right].
\end{align*}
Using \eqref{6} we obtain
\begin{equation*}
\phi(z)= \exp\left( \zeta z + \frac{1}{2}\epsilon(t, T) \sigma_t^2(z^2+ 2\beta z) \right) E^{\mathbb{P}} \left[\exp\left( \int_t^T \left(\rho z + \frac{1}{2}(z^2+ 2 \beta z) \epsilon(s, T) \right)\, dZ^{(e)}_{\lambda s} \right) \right].
\end{equation*}

Clearly if $z \in \mathcal{S}$, then $\Re (\rho z + \frac{1}{2}(z^2+ 2 \beta z) < \hat{\theta}$. Thus the result follows from \eqref{8}. 
\end{proof}

Proof of Theorem \ref{jumpgen}:
\begin{proof}
Since $z$ is a L\'evy process with characteristics $(\mu, \sigma^2, \nu)$ under $P_0$ and characteristics $(\mu, \sigma^2, (1+ax)\nu)$ under $P_1$, we apply the generalized Girsanov's Theorem. Using $\beta$ as in \eqref{beta}, we obtain
$$\frac{dP_0}{dP_{1}}=\mathcal{E}\left(-N.\right)_t,$$
where
$$N_t=\beta W_t +\int_0^t\int_{x>0} ax (J-\nu)(ds,dx),$$
$aJ$ is the jump measure for $N$, $W$ is a standard Brownian motion, and $\mathcal{E}$ is the Dol\'eans-Dade exponential, defined below. This gives that $N_t$ is a L\'evy process with characteristics $$(  m, \beta^2,-a\nu).$$ Then, by \cite{Tank} (Proposition 8), we obtain characteristics $$(\gamma,\beta^2,-K),$$ for $u_t$. Finally, by \cite{Tank, AS}, the process has the stated generator.\end{proof}

Proof of Lemma \ref{EUlem}:
\begin{proof}
First, consider (A1): 
$$F(x,u,p,X,i_1)- F(x,u,p,Y,i_2) =  \frac{\beta}{2}(Y-X) +i_2- i_1 \geq0$$
if $i_2\leq i_1$ and $Y\leq X$.

Next, $F(x,u,p,X,i)- F(x,v,p,X,i)=M(u-v)$, so choosing $\gamma=M >0$, we have property (A2).

Property (A3) is satisfied because $F$ is linear in each argument, and (A4) is satisfied because $F$ does not depend on its first argument explicitly. Last, $K$ is a L\'evy-It\^o measure by the assumptions of the underlying L\'evy process. 

\end{proof}

Proof of Theorem \ref{solution}:
\begin{proof}
We define
\begin{align*}
H(x)&=\int_{0}^{\infty} \xi(x+y) K(dy), \nonumber \\
M&=\int_0^\infty K(dy).
\end{align*}
Consequently,
\begin{align}
0=\mathcal{L}\xi(x)&= -\gamma \xi'(x) +\frac{1}{2} \beta^2 \xi''(x) +\int_{\mathbb{R}_+} \left(\xi(x+y)-\xi(x)-\frac{y \xi'(x)}{1+|y|}\right)K(dy)\nonumber
\end{align}
can be rewritten as
\begin{align}
0=& -\gamma \xi'(x) +\frac{1}{2} \beta^2 \xi''(x) +H(x)-M \xi(x)-C \xi'(x). \nonumber
\end{align}

The sign on $H$ is positive; therefore, we have sub-solution equation
$$0= \frac{1}{2} \beta^2 \xi''(x) -(C+\gamma) \xi'(x) -M \xi(x).$$
On the other hand, since $\xi$ is monotonic and positive inside $R$,
$$\xi(x+y)-\xi(x)\leq 0 \iff H(x)-M\xi(x)\leq 0.$$
Using this, we have super-solution equation
$$0= -\gamma \xi'(x) +\frac{1}{2} \beta^2 \xi''(x) -C \xi'(x). $$

Finally, applying the previous theorem \ref{EUtheorem}, we have the existence of a viscosity solution.
\end{proof}

\textbf{Acknowledgment}: The authors would like to thank the anonymous reviewers for their careful reading of the manuscript and for suggesting points to improve the quality of the paper.

\end{document}